\documentclass[sigconf]{acmart}

\setcopyright{none}
\renewcommand\footnotetextcopyrightpermission[1]{}

\usepackage{fancyhdr}
\setcopyright{none}
\renewcommand\footnotetextcopyrightpermission[1]{}

\usepackage{graphicx} 
\usepackage{algorithm}
\usepackage{algpseudocode}
\usepackage{amsthm}
\usepackage{todonotes}
\usepackage{bm}

\newtheorem{lemma}{Lemma}
\newtheorem{theorem}{Theorem}

\usepackage{comment}

\usepackage{xcolor}
\usepackage{xspace}
\usepackage{tikz}
\usetikzlibrary{arrows.meta,positioning,calc,shapes.geometric, shapes, snakes}

\usepackage{makecell}
\usepackage{cleveref}

\newcommand{\mhp}{
  \textsc{Multihop}\,-\!
  \raisebox{-0.5ex}{\includegraphics[height=2.9ex]{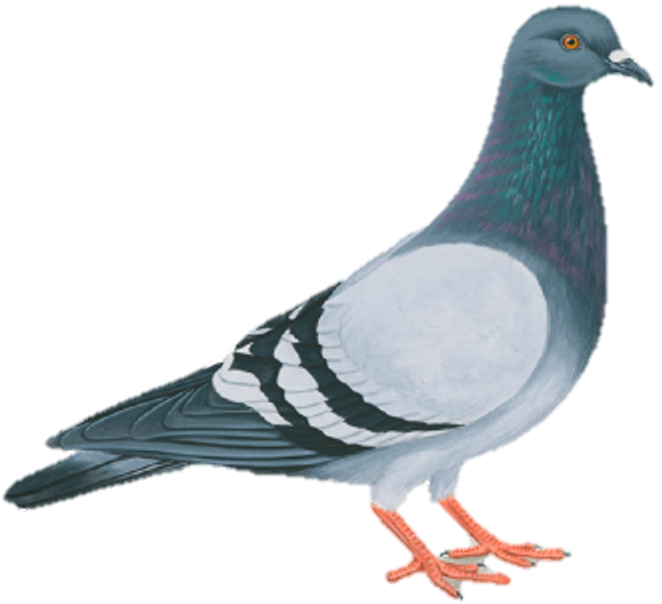}}\xspace}

\newcommand{\thp}{\textsc{2-hop}\,-\!
  \raisebox{-0.5ex}{\includegraphics[height=2.9ex]{Pigeon.png}}\xspace}
\newcommand{\shp}{\textsc{Singlehop}\,-\!
  \raisebox{-0.5ex}{\includegraphics[height=2.9ex]{Pigeon.png}}\xspace}
\newcommand{\pigeononedge}[3][0.55]{%
  \path (#2) -- (#3) coordinate[pos=#1] (pigeonpos);
  \path let \p1=(#2), \p2=(#3) in
    node[
      at={(pigeonpos)}
    ] {\includegraphics[width=6mm]{pigeon.png}};
}

\newcommand{\specialstar}{%
  \tikz{\node[star,star point ratio=0.5,rotate=180, color=yellow, fill=yellow,draw,inner sep=3pt] at(0,-0.1) {};}\xspace}
\usepackage{balance}

\begin{document}

\title[Algorithms for Carrier Pigeons]{The Carrier Pigeon Internet Protocol:\\An Algorithmic 
(and Lighthearted) Perspective}

\author{Matthias Bentert}
\affiliation{%
  \institution{TU Berlin}
  \city{Berlin}
  \country{Germany}
}
\email{matthias@bentert.de}

\author{Shay Kutten}
\affiliation{%
  \institution{Technion}
  \city{Haifa}
  \country{Israel}
}
\email{kutten@technion.ac.il}

\author{Darya Melnyk}
\affiliation{%
  \institution{TU Berlin}
  \city{Berlin}
  \country{Germany}
}
\email{melnyk@tu-berlin.de}

\author{Tijana Milentijevi\'c}
\affiliation{%
  \institution{TU Berlin}
  \city{Berlin}
  \country{Germany}
}
\email{tijana.milentijevic@tu-berlin.de}

\author{Stefan Schmid}
\affiliation{%
  \institution{TU Berlin}
  \city{Berlin}
  \country{Germany}
}
\email{stefan.schmid@tu-berlin.de}

\begin{abstract}
	The theoretical model behind the pigeon post as a link layer in a communication network was introduced by Shannon (under the guise of studying One-Time Pads for cryptography). That is, to send a one-hop message to $v$, a node $u$ needs a mail pigeon bred and raised at $v$. When sending a message using a pigeon to $v$, node $u$ loses the pigeon. To send another message to $v$, node $u$ needs another pigeon of $v$.
It has been demonstrated that the communication bandwidth achievable with pigeon post can exceed that of networks using other media. This has already motivated the introduction of Internet standards that allow the use of pigeons as Internet link-layer media.

In this paper, we begin to fill in the missing piece: designing algorithms for breeding and scheduling pigeons to meet a given communication demand efficiently, minimizing the number of pigeons required.
We consider singlehop, 2-hop, and multihop pigeon use. While the singlehop variant admits a simple characterization, both the 2-hop and the multihop variants are NP-hard. For the latter variants, we present a polynomial-time algorithm based on demand aggregation that achieves a 2-approximation for the number of pigeons used.
We believe that this pigeon-based perspective offers both amusing and instructive insights into network design and hopefully, into ornithology.
\end{abstract}



\maketitle

\section{Introduction}\label{sec:intro}
The AI revolution places heavy demands on communication bandwidth, which is often the bottleneck in AI tasks such as training models; see, e.g., \cite{narayanan2021efficient,huang2019gpipe}.  
In the quest to increase bandwidth, researchers turn to the time-tested communication method of the pigeon post, but adapt it for modern communication. See IP over Avian Carriers (IPoAC)  \cite{rfc1149,rfc6249} and the Carrier Pigeon Internet Protocol (CPIP) \cite{cpip}. A motivation for such link layer standards can be the demonstration that the bandwidth of communication using pigeons could be much higher than that of competing media\footnote{Provided that one does not communicate at night and that one does not mind the shit \cite{pigeon-experiment}.}, since a pigeon can carry a chip containing a substantial amount of information by using the PEI protocol (Pigeon Enabled Internet) in the TCP framework (Transmission by Carrier Pigeons) \cite{pigeon-experiment}. This was demonstrated again in \cite{bbc, mail-gardian} and in \cite{BBC2,australia}. Pigeons have also been used in other layers of the Internet, see, e.g. \cite{PigeonRank,wiki:google}.

Environmental aspects also support the use of pigeon post. Fiber optics, the current leading alternative, requires substantial amounts of silicon, commonly in the form of sand. This comes at a time when sand in various parts of the world is disappearing, a development that worries environmentalists, economists, and governments  \cite{cao2022material,salle2022global,lamb2023constructing}. The UN is also worried \cite{peduzzi2022sand}.
So, next time you are sitting on the beach in the Riviera,
or Cancun, or Thailand, or the Maldives, remember that the sand is in danger. If you want to keep having such vacations, supporting Pigeon Post as opposed to fiber optics is the way to go.

In the long term, pigeon post also has the potential to help increase agricultural yields.
Using more pigeons generates more guano \cite{pigeon-experiment}.
This guano is a good fertilizer~\cite{singh2019nutrient,masin2024composting}, which suggests that using pigeon post may help grow more food. This, in turn, may allow one to free up land for ecological restoration, thus restoring forests and wetlands, and increasing biodiversity. \footnote{On a serious note, whatever one may think of pigeon post, the environmental crisis~\cite{unep25} and the hunger crisis~\cite{ifad2017state} are prevalent; let us hope that this paper will make at least a small contribution to increasing awareness.}

The theoretical basis for Pigeon Post was laid in Shannon's seminal work~\cite{shannon1949communication}. Shannon, however, did not mention pigeons by name. Instead, he referred to his version of Vernam's encryption protocol~\cite{vernam1926cipher}. This is known today as a ``one-time pad'' and appears in many popular spy-related books in the fiction literature, e.g.~\cite{havana, crypto}.\footnote{A reader who reads these two books as a result of our paper has already gained a lot from the paper. {The authors do not have any financial interest in the books or the publishers.}} The spy has some text (or a bit string) given to her by her country's secret service. Let us call this text a ``pigeon''. Using this pigeon, she encrypts a message to generate a ``ciphertext'' to be read by the secret service (often, the encryption is just an XOR operation between each bit of the message and the corresponding bit of the pigeon). The service reads the message by applying the reverse operation (often just XORing the ciphertext with the corresponding bit of the same pigeon). Shannon analyzed the version in which both sides then discard the pigeon. Hence, if the spy initially has $x$ pigeons to send messages to the secret service (each such pigeon could not be sent to anyone but that secret service), then after the message is sent, the spy has only $x-1$ pigeons remaining. 

As is sometimes regrettably the case, the ancient Egyptians, Persians, Greeks, Romans, etc., did not wait for Shannon's theoretical foundations but instead proceeded with the practice of using pigeons to send messages.\footnote{Modern physics does not rule out the possibility of time travel, so it may be the case that those early adopters did rely on Shannon's work after all \cite{Luminet_2021}.} See \cite{wiki:post}.

\paragraph*{Our Contributions} 

We view this paper as an ``Introduction to Pigeon Post''. We 
first formally define how network communication can be implemented using mail pigeons (also known as carrier or homing pigeon). We thereby assume that the mail pigeons can be bred in a demand-aware manner, i.e., the breeders are aware of the future communication demands. Optimizing these breeding locations is critical as later in their lives, mail pigeons can only fly in one direction: \emph{home} (their birthplace). 
The objective is to find a pigeon traffic pattern that uses as few pigeons as possible (model details will follow).

We first analyze a scenario where messages have to be sent \emph{directly} to their receivers, by a single pigeon (singlehop), and present a simple optimal algorithm. We then generalize the model and allow forwarding of messages using two or more homes; these homes can serve as intermediate nodes where messages are relayed to other pigeons (multihop). We show that the general problem is NP-hard, even if the number of homes for forwarding is unlimited. 

We further present a 2-approximation algorithm for a scenario where we are allowed to use one intermediate node, that is, two pigeons (2-hop). We also present elegant Integer Linear Programs for the NP-hard problems we discuss. Our results are summarized in Table~\ref{tab:results}.

\begin{table}[t] 
    \centering
    \resizebox{\columnwidth}{!}{%
    \begin{tabular}{lccc}
    \hline
    Model & Approximation & Number of pigeons & Runtime \\ \hline
    \shp 
    & \makecell[c]{optimal \\ \textcolor{gray}{Algorithm~\ref{alg:singlehop-trivial}}}
    & \makecell[c]{$\Theta(n^2)$ \\ \textcolor{gray}{Theorem~\ref{thm:singlehop}}}
    & \makecell[c]{$O(|S|+|D|)$ \\ \textcolor{gray}{Theorem~\ref{thm:singlehop}}} \\ \hline

    \thp 
    & \makecell[c]{$2$-approx. \\ \textcolor{gray}{Algorithm~\ref{alg:twohop-coordinator}}}
    & \makecell[c]{$O(|S|+|D|)$ \\ \textcolor{gray}{Theorem~\ref{thm:twohop}}}
    & \makecell[c]{$O(n^2)$\\ \textcolor{gray}{Theorem~\ref{thm:twohop}}} \\ \hline

    \thp{} ILP 
    & \makecell[c]{optimal}
    & \makecell[c]{optimal}
    & \makecell[c]{exponential \\ \textcolor{gray}{Theorem \ref{thm:ILPtp}}} \\ \hline

    \mhp 
    & \makecell[c]{$2$-approx. \\ \textcolor{gray}{Algorithm~\ref{alg:twohop-coordinator}}}
    & \makecell[c]{$O(|S|+|D|)$\\ \textcolor{gray}{Theorem~\ref{thm:twohop}}}
    & \makecell[c]{$O(n^2)$ \\ \textcolor{gray}{Theorem~\ref{thm:twohop}}} \\ \hline

    \mhp{} ILP 
    & \makecell[c]{optimal}
    & \makecell[c]{optimal}
    & \makecell[c]{exponential \\ \textcolor{gray}{Theorem \ref{thm:ILPmp}}} \\ \hline
    \end{tabular}
    }
    \caption{Summary of results.}
    \label{tab:results}
\end{table}

\paragraph*{Further Related Work} 

The design of demand-aware structures such as codes (e.g., Huffman codes~\cite{huffman2007method}), data structures (e.g., biased binary search trees~\cite{demaine2007dynamic,sleator1985self,bent1985biased}), and networks (e.g., splay nets~\cite{schmid2015splaynet,avin2020demand}) is an evergreen topic in algorithm theory. 
Recently, demand-aware networks have received particular interest in the context of reconfigurable datacenter networks, whose topology can be optimized to match the traffic workload~\cite{avin2025revolutionizing,farrington2010helios,ghobadi2016projector}, see, for example, Google's Jupiter datacenter~\cite{poutievski2022jupiter}. Demand-aware networks are attractive as communication traffic is known to exhibit substantial structure~\cite{avin2020complexity}, which can be exploited for optimization.

Many optimization problems in communication rely on demand matrices, including traffic engineering~\cite{fortz2000internet}, and the early works date back to the beginnings of the Internet~\cite{Kleinrock1976}. There is also interesting research on estimating the amount of communication to be sent from each node $i$ to each node $j$ \cite{MedinaTSBD02}. Matrices for a directed graph (such as the graph we use here) are addressed, e.g., in \cite{HajiaghayiKLR05}. For even older versions, see e.g., \cite{Hitchcock1941} (the famous transportation problem) and \cite{BeckmannMcGuireWinsten1956}.

The problem is also connected to multi-commodity flow problems \cite{greening2025strengthening} as well as graph layout problems \cite{even1975complexity,diaz2002survey}: the design of a demand-aware graph of degree 2 corresponds to the minimum linear arrangement problem~\cite{avin2020demand}. These problems are NP-hard in many scenarios, also on directed networks \cite{even1975complexity} (like our demand graph). Vehicle routing problem~\cite{toth2002vehicle} relates to our problem as well. However, different variants of this problem, including the multi-depot vehicle routing~\cite{crevier2007multi}, to the best of our knowledge, do not consider the relaying of goods between vehicles, which is essential for our problem.

More remotely, our storage model is inspired by the Pigeon-Hole Principle (see, e.g., \cite{ajtai1994complexity}): that is, if there are more pigeons than holes in a pigeon house, then some holes will simply need to accommodate multiple pigeons.

\section{Model}

In our model, communication demands may arise between different locations. We represent these locations by a set of nodes $V$ with $|V| = n$ and the communication demand as a directed and unweighted \emph{demand graph} $G_D = (V, E_D)$. There is a directed edge $(i,j) \in V \times V$ with $i \neq j$ in $G_D$, if there is a communication demand from node $i$ to node $j$. We assume that the demand graph is stored in an adjacency matrix $M_D$, which we refer to as the demand matrix. Note that $M_D$ is not necessarily symmetric and satisfies $M_D(i,i) = 0$ for all $i\in V$. Nodes with outgoing demand are referred to as sources $s \in S$, and nodes with incoming demand as destinations $d\in D$; note that a node may play both roles.

The communication infrastructure is induced by pigeons and can be viewed as a dynamic directed multigraph, called \emph{infrastructure graph}. Let $P$ denote a set of pigeons (at a given time). Each pigeon $p\in P$ is associated with two nodes: a home node $h(p) \in V$, where the pigeon was bred and (only) to which it can fly, and a remote node $r(p) \in V$, where the pigeon is initially placed. As we will see, our problems concern defining home and remote nodes so that a minimal number of pigeons can serve a given communication demand.
Upon release from its remote location, pigeon $p$ flies directly from $r(p)$ to $h(p)$ without intermediate stops.
A pigeon thus induces a directed edge $(r(p), h(p))$ in the infrastructure graph. 
When $p$ flies, the edge is deleted.
Multiple pigeons may induce parallel edges between the same pair of nodes, each used at a different time, to ensure that transportation demand is satisfied.

The communication demand does not need to be carried directly from its source to the destination by a single pigeon. Instead, demand may be routed along directed paths in the infrastructure graph. A unit of demand originating at node $i$ and destined for node $j$ may traverse a sequence of intermediate nodes. At intermediate nodes, demand can be gathered from arriving pigeons and split to the corresponding departing pigeons. Similarly, demand from multiple arriving pigeons can be batched together and sent using a single pigeon.


We distinguish between singlehop, 2-hop, and multihop uses of pigeons. In the \shp{} problem variant, each unit of demand must be transported directly from its source to its destination by a single pigeon, without using intermediate nodes. Consequently, no aggregation or forwarding of demand is allowed, and a pigeon flying from node $i$ to node $j$ can only carry demand destined for $j$ that originates at $i$. In the \thp{} problem, the demand from node $i$ to node $j$ may either be transported directly or routed via a single intermediate node $l$, resulting in a path $i \to l \to j$ in the infrastructure graph. 

The \mhp{} problem is a generalization of the \thp, in which communication demand may be forwarded through multiple intermediate nodes via multiple pigeons sequentially.  

In this paper, we assume that pigeons can carry arbitrary amounts of information. When a pigeon $p$ arrives at its home node $h(p)$, all the demand it carries is delivered to the home node. If $h(p)$ is not the final destination of the demand, the information can be stored at that node and forwarded at a later time by pigeon $p^\prime$, which has $h(p)$ as its remote node, i.e., $h(p)=r(p^\prime)$. 
In that case, $h(p)$ is an intermediate node for the demand.  
Demand may wait arbitrarily long at nodes before being forwarded further. Nodes have unlimited storage capacity and may aggregate incoming demand without restriction. This part of the model actually follows from the Pigeon-Hole Principle: that is, if there are more pigeons than holes at pigeon house $h(p)$, then some holes simply will need to accommodate multiple pigeons. Figure~\ref{fig:bipartite-to-solution-existing-nodes} illustrates a demand graph and a corresponding infrastructure graph induced by pigeons in a \mhp{} problem.

The demand is delivered successfully if there exists an assignment of all demands to pigeon flights over time such that, for every demand pair $(i,j) \in G_D$, the demand 
is routed from $i$ to $j$ along a directed path in the (dynamic) infrastructure graph. 
In general, for the three problems \shp, \thp{} and \mhp{}, the pigeon post operates in two conceptual phases. In the first, offline planning phase, the complete demand matrix $M_D$ is known, and our objective is to breed pigeons (i.e., define their home nodes) and place them in their remote nodes according to our network design decisions. In the second phase, the execution phase, pigeons fly to their home nodes according to the required schedule (i.e., wait for the predecessor pigeon in multihop routing with intermediate nodes). The demand is routed along the established infrastructure graph following a predefined schedule. 

Our objective is to construct an infrastructure graph that routes all demands from the demand matrix while minimizing the total number of pigeons. Besides the minimization problems, we also consider decision versions of these problems (e.g., when studying hardness), where we ask whether a given demand can be satisfied with at most $k$ pigeons.

\begin{figure}[t]
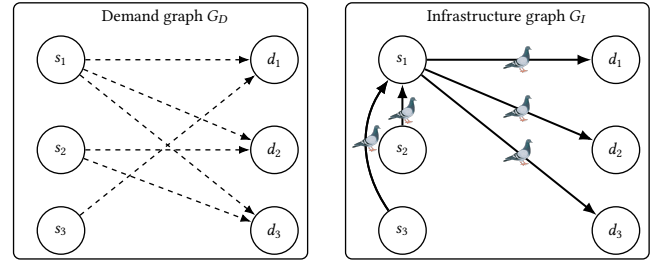

\centering
\resizebox{\columnwidth}{!}{%
\begin{tikzpicture}[
  >=Latex,
  node/.style={circle, draw, thick, minimum size=10mm, inner sep=0pt, font=\large},
  demand/.style={-Latex, thick, dashed},
  pigeon/.style={-Latex, very thick},
  flow/.style={-Latex, thick, opacity=0.55},
  lab/.style={font=\large, inner sep=1pt},
  box/.style={draw, rounded corners, thick}
]

\begin{scope}
  \coordinate (Lsw) at (-1.0,-3.2);
  \coordinate (Lne) at (5.2,2.2);
  \draw[box] (Lsw) rectangle (Lne);
  \node[lab] at (2.2,1.9) {Demand graph $G_D$};

  \node[node] (s1L) at (0,1.0) {$s_1$};
  \node[node] (s2L) at (0,-0.9) {$s_2$};
  \node[node] (s3L) at (0,-2.6) {$s_3$};

  \node[node] (t1L) at (4.5,1.0) {$d_1$};
  \node[node] (t2L) at (4.5,-0.9) {$d_2$};
  \node[node] (t3L) at (4.5,-2.6) {$d_3$};

  \draw[demand] (s1L) -- node[lab, above] {} (t1L);
  \draw[demand] (s1L) -- node[lab] {} (t2L);
  \draw[demand] (s1L) -- node[lab, below] {} (t3L);

  \draw[demand] (s2L) -- node[lab] {} (t2L);
  \draw[demand] (s2L) -- node[lab, below] {} (t3L);

  \draw[demand] (s3L) -- node[lab, above] {} (t1L);

\end{scope}

\begin{scope}[xshift=7.2cm]
  \coordinate (Rsw) at (-1.2,-3.2);
  \coordinate (Rne) at (5.2,2.2);
  \draw[box] (Rsw) rectangle (Rne);
  \node[lab] at (2.2,1.9) {Infrastructure graph $G_I$};

  \node[node] (s1) at (0,1.0) {$s_1$};
  \node[node] (s2) at (0,-0.9) {$s_2$};
  \node[node] (s3) at (0,-2.6) {$s_3$};

  \node[node] (t1) at (4.5,1.0) {$d_1$};
  \node[node] (t2) at (4.5,-0.9) {$d_2$};
  \node[node] (t3) at (4.5,-2.6) {$d_3$};

  \draw[pigeon] (s3) to[bend left=35] (s1);
  \draw[pigeon] (s1) -- (t1);
  \draw[pigeon] (s1) -- (t2);
  \draw[pigeon] (s1) -- (t3);

  \draw[pigeon]
  (s3) to[bend left=35]
  node[pos=0.55] {\includegraphics[width=6mm]{pigeon.png}}
  (s1);
  \draw[pigeon]
  (s2) to[bend left=0.1]
  node[pos=0.35] {\includegraphics[width=6mm]{pigeon.png}}
  (s1);
  \pigeononedge{s1}{t1}
  \pigeononedge{s1}{t2}
  \pigeononedge{s1}{t3}

\end{scope}

\end{tikzpicture}
}
\caption{On the left side a demand graph $G_D$ is depicted with 3 source and 3 destination nodes. The right side shows the corresponding infrastructure graph $G_I$	induced by pigeons, illustrating an optimal placement and routing solution that satisfies all the demands while minimizing the total number of pigeons used. In this example, a pigeon $p$ flying from $s_3$ to $s_1$ has $s_1$ as its birthplace, i.e. $h(p)=s_1$. The pigeon is brought to a remote node $s_3$, i.e. $r(p)=s_3$, and flies home directly to $s_1$. }
\label{fig:bipartite-to-solution-existing-nodes}
\end{figure}

\section{Theoretical Results}

In this section, we present theoretical results for the three problem variants: \shp, \thp{}, and \mhp{}. For all three problems, we present a simple lower bound on the number of pigeons needed: 
\begin{theorem}[Lower bound on the number of pigeons] \label{obs:min-pigeons}
    The minimal amount of pigeons needed to transfer the entire demand is $|P| \geq \max(|S|, |D|)$, where $P$ denotes the set of pigeons, $S$ and $D$ a set of demand sources and demand destinations, respectively. 
\end{theorem}
\begin{proof}
    Observe that for each edge $(i,j)$ in the demand graph, there must be at least one pigeon starting in the remote node $i$. Similarly, there must be at least one pigeon arriving in its home $j$. The lower bound is achieved by summing over all remote nodes and home nodes, respectively, and choosing the maximum value. 
\end{proof}

\subsection{\shp{} Solution} 
We begin by considering the singlehop pigeons (\shp) problem, in which each unit of demand must be transported directly from its source to its destination by pigeons. In contrast to the multihop setting, intermediate nodes are not permitted, and demand cannot be aggregated or forwarded through other nodes. Consequently, each pigeon can only carry demand originating at its remote node and destined for its home node.

Algorithm \ref{alg:singlehop-trivial} shows a trivial solution in the \shp{} setting, in which each source node sends a pigeon directly to the corresponding destination. In particular, the problem reduces to selecting a pigeon for each edge in the demand graph. Observe that the optimal solution can be computed efficiently, as it only needs to consider all sources and destinations. In the worst case, however, all $n^2$ demand edges are present in the graph. These results are summarized in the following theorem: 

\begin{theorem}\label{thm:singlehop}
    Algorithm~\ref{alg:singlehop-trivial} computes the optimal amount of pigeons in $O(|S|+|D|)$ time. This algorithm uses $\Theta(n^2)$ pigeons in the worst case. 
\end{theorem}

\begin{algorithm}[t]
\caption{\shp{} Algorithm}
\label{alg:singlehop-trivial}
\begin{algorithmic}[1]
\State $P \gets \emptyset$
\Statex \textit{Offline Planning Phase}
\ForAll{directed edges $(i,j) \in V \times V$ with $i \neq j$}
    \If{$(i,j)\in G_D$}
            \State breed pigeon $p$ at node $j$ 
            \State $r(p) \gets i$, $h(p) \gets j$ \Comment{set remote and home node}
            \State ship pigeon $p$ to node $i$
            \State $P \gets P \cup \{p\}$
    \EndIf
\EndFor
\Statex \textit{Flight Scheduling Phase}
\State let all pigeons fly simultaneously to their respective homes
\end{algorithmic}
\end{algorithm}

\subsection{\thp{} Solution} 

We now move beyond direct communication and study the use of intermediate nodes. Allowing demand to be forwarded through other nodes significantly increases the capabilities of the pigeon post by enabling mail aggregation before it is delivered to its final destinations. We first focus on the \thp{} problem, in which each demand is routed using at most two pigeon flights. The \thp{} problem can be viewed as a restricted form of the general \mhp{} problem, in which all mail delivery paths are required to have a length of at most two.

Algorithm~\ref{alg:twohop-coordinator} presents the \thp{} coordinator algorithm. The coordinator algorithm is based on gathering all the demand at one node, called the coordinator. Then, the coordinator spreads the demand and sends pigeons to the corresponding destinations. Note that we consider each weakly connected component separately in the algorithm, and pick the node with the largest degree in the connected component as the coordinator. Any demand in a connected component is then carried from one source to the coordinator (with one pigeon) and then from the coordinator to the destination (using a single pigeon for each destination). 

\begin{algorithm}[t]
\caption{\thp{} Coordinator Algorithm}
\label{alg:twohop-coordinator}
\begin{algorithmic}[1]
\State $P \gets \emptyset$
\State Execute the following algorithm for each connected component of $G_D$
\State coordinator $K \gets $ highest degree node of the current connected component
\Statex \textit{Offline Planning Phase}
\ForAll{nodes $i \in V$ with $i \neq K$}
    \If{$\sum_{j \in V} M_D[i,j] > 0$}
        \State breed pigeon $p$
        \State $r(p) \gets i$, $h(p) \gets K$ \Comment{gather all outgoing demand of $i$ at coordinator}
        \State ship pigeon $p$ to node $i$
        \State $P \gets P \cup \{p\}$
    \EndIf
\EndFor
\ForAll{nodes $j \in V$ with $j \neq c$}
    \If{$\sum_{i \in V} M_D[i,j] > 0$}
        \State breed pigeon $p$
        \State $r(p) \gets K$, $h(p) \gets j$ \Comment{send all demand destined for $j$ from coordinator}
        \State ship pigeon $p$ to node $K$
        \State $P \gets P \cup \{p\}$
    \EndIf
\EndFor
\Statex \textit{Flight Scheduling Phase}
\State let pigeons with a home in $K$ fly to $K$
\State let pigeons with a home not in $K$ fly from $K$ to their homes
\end{algorithmic}
\end{algorithm}

\begin{theorem}\label{thm:twohop}
    Let $c$ denote the number of weakly connected components in $G_D$, and ${C_1,\ldots,C_c}$ be the corresponding weakly connected components. We use $\Delta(C_i)$ to denote the maximum degree of the weakly connected component $C_i$. Algorithm~\ref{alg:twohop-coordinator} computes a solution with up to $|S| + |D| - \sum_{i\in[c]}\Delta(C_i)$ pigeons. It is thus a ($2-\sum_{i\in[c]}\Delta(C_i)/n$)-approximation of the optimal \thp{} solution. The algorithm runs in $O(n^2)$ time.
\end{theorem}

\begin{proof}
    In each weakly connected component, Algorithm~\ref{alg:twohop-coordinator} selects a coordinator node and routes the whole demand in two phases: first, all sources send their entire demand to the coordinator, and second, the coordinator sends the accumulated demand to the corresponding destinations. In each weakly connected component, the algorithm correctly serves the demand. 
    
    This solution requires at most one pigeon from each source in $S$ to the respective coordinator and at most one pigeon from the coordinator to each destination in $D$, resulting in a cost of at most $|S| + |D|$ pigeons. Since the coordinator is chosen to be the highest degree node in the weakly connected component, i.e., it is a source and/or a destination, the solution saves at least as many pigeons as $\Delta(C_i)$ for each weakly connected component $C_i, i\in[c]$. Let $ALG$ denote the total cost of the algorithm. Then, $ALG \leq |S| + |D| -\sum_{i\in[c]}\Delta(C_i)$.

    In Theorem~\ref{obs:min-pigeons}, we derived a lower bound of $\max(|S|, |D|)$ on the number of pigeons for an optimal solution $OPT$. The approximation ratio of Algorithm~\ref{alg:twohop-coordinator} can be computed as 
    \begin{align*}ALG \leq |S| + |D| - \sum_{i\in[c]}\Delta(C_i) &\le 2 \cdot \max(|S|, |D|) - \sum_{i\in[c]}\Delta(C_i) \\
    &\leq (2-\sum_{i\in[c]}\Delta(C_i)/n) \cdot OPT.
    \end{align*}
    Hence, the algorithm computes a $(2-\sum_{i\in[c]}\Delta(C_i)/n)$-approximation of the optimal solution.

    The computed approximation factor for Algorithm~\ref{alg:twohop-coordinator} is tight. Consider for example a cyclic demand instance on $n$ nodes, where each node has exactly one outgoing demand and exactly one incoming demand. There is only one connected component in the demand graph, and every node has degree $2$. In this case, the optimal solution requires $n$ pigeons, while the Algorithm~\ref{alg:twohop-coordinator} uses $2\cdot n-2$ pigeons, resulting in an approximation factor $2-2/n$.

    The runtime of the algorithm is dominated by the computation of the weakly connected components which can be done in $O(n^2)$ time. 
\end{proof}

In the following, we show that it is hard to compute an optimal solution for the \thp{} problem.

\begin{theorem}
\label{thm:red-2spanner-to-thp}
The decision variant of \thp{} is NP-hard. 
\end{theorem}

\begin{proof}
We reduce from \textsc{3SAT}. 
An instance of \textsc{3SAT} consists of a Boolean formula $\varphi$ in conjunctive normal form, where each clause contains exactly three literals.

Given an instance $\varphi$ of \textsc{3SAT}, we construct an instance $(G_D,k)$ of \thp{} as follows. The demand graph $G_D=(V,E)$ consists of one node for each clause $C_1,\dots, C_m$; one node for each positive literal $x_1, \dots, x_n$ and negative literal $\overline{x}_1, \dots, \overline{x}_n$; one special node \specialstar and additional $(6n+12)(2n+3m)$ nodes, which will be explained later. We set~$k=12n^2+18nm+27n+39m$.

We place the demand edges as follows: from each clause $C_i$ to each literal $x_j$ or~$\overline{x}_j$ that appears in that clause; from each clause $C_i$ to the special node~\specialstar; from each literal to~\specialstar; from $x_j$ to $\overline{x}_j$ and from $\overline{x}_j$ to $x_j$ for each~$j\in[n]$. 
Demand edges from each clause $C_i$ to each literal $x_j$ or~$\overline{x}_j$ that appears in that clause and between literals are called \emph{forced edges}. In total, there are $2n+3m$ forced edges. For each forced edge~$e=(a,b)$, we add $(6n+12)$ additional nodes~$u_{i,j}^e$ with~$i \in [2n+4]$ and~$j \in [3]$.
We also add demand edges~$(u_{i,1}^e, u_{i,2}^e)$, $(u_{i,1}^e,u_{i,3}^e)$, $(u_{i,2}^e,u_{i,3}^e)$, $(u_{i,2}^e,a)$, $(u_{i,3}^e,a)$, and~$(u_{i,3}^e,b)$ for each forced edge~$e=(a,b)$ and each~$i \in [2n+4]$.
For each forced edge~$e$ and each~$i\in [2n+4]$, we call the set~$\{u_{i,1}^e,u_{i,2}^e,u_{i,e}^e\}$ an \emph{arm} of the forced edge gadget for~$e$.
An example of the construction is shown in Figure~\ref{fig:clause-gadget} and the forced edge gadget is illustrated in Figure~\ref{fig:forced-edge-gadget}. Since the construction can clearly be computed in polynomial time, it remains to show correctness.

\begin{figure}[t]
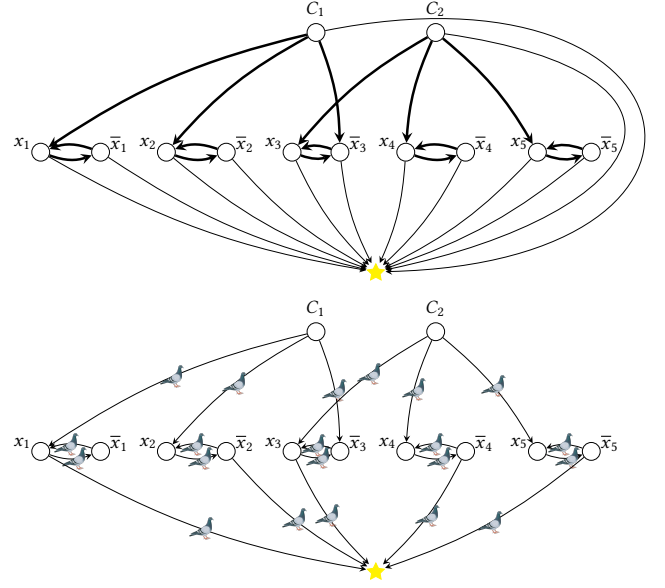

\centering
\resizebox{1.2\columnwidth}{!}{%
\begin{tikzpicture}[> = stealth]
\begin{scope}[rotate=-90]

  \node[circle,draw,inner sep=3pt,label=$C_1$] at(0,0) (C) {};
\node[circle,draw,inner sep=3pt,label=$C_2$] at(0,2) (C2) {};

\node[circle,draw,inner sep=3pt,label={[yshift=4pt, xshift=3pt]left:${x_5}$}] at(2,3.7) (x5) {};
\node[circle,draw,inner sep=3pt,label={[yshift=4pt, xshift=3pt]left:${x_4}$}] at(2,1.5) (x4) {};
  \node[circle,draw,inner sep=3pt,label={[yshift=4pt, xshift=3pt]left:${x_3}$}] at(2,-0.4) (x3) {};
  \node[circle,draw,inner sep=3pt, label={[yshift=4pt, xshift=3pt]left:${x_2}$}] at (2,-2.5) (x2) {};
  \node[circle,draw,inner sep=3pt,label={[yshift=4pt, xshift=3pt]left:$x_1$}] at(2,-4.6) (x1) {};

  \node[circle,draw,inner sep=3pt,label={[xshift=-3pt, yshift=4pt]right:$\overline{x}_5$}] at(2,4.6) (nx5) {};
  \node[circle,draw,inner sep=3pt,label={[xshift=-3pt, yshift=4pt]right:$\overline{x}_4$}] at(2,2.5) (nx4) {};
  \node[circle,draw,inner sep=3pt,label={[xshift=-3pt, yshift=4pt]right:$\overline{x}_3$}] at(2,0.4) (nx3) {};
  \node[circle,draw,inner sep=3pt,label={[xshift=-3pt, yshift=4pt]right:$\overline{x}_2$}] at(2,-1.5) (nx2) {};
  \node[circle,draw,inner sep=3pt,label={[xshift=-3pt, yshift=4pt]right:$\overline{x}_1$}] at(2,-3.6) (nx1) {};

  \node[circle,inner sep=-1pt] at(4,1) (S) {\specialstar};

  \draw[->, very thick] (C)to[bend right=10] (x1);
  \draw[->, very thick] (C) to[bend right=10]  (x2);
  \draw[->, very thick] (C) to[bend left=10] (nx3);
  \draw[->, very thick] (C2)to[bend right=10] (x3);
  \draw[->, very thick] (C2)to[bend right=10] (x4);
  \draw[->, very thick] (C2) to[bend left=10](x5);

  \draw[->] (C) to[out=100, in=-270, looseness=4.1] (S);
  \draw[->] (C2) to[out=80, in=-260, looseness=3.0] (S);

  \draw[->, very thick] (x1) to[bend right=20] (nx1);
  \draw[->, very thick] (nx1) to[bend right=20] (x1);
  \draw[->, very thick] (x2) to[bend right=20] (nx2);
  \draw[->, very thick] (nx2) to[bend right=20] (x2);
  \draw[->, very thick] (x3) to[bend right=20] (nx3);
  \draw[->, very thick] (nx3) to[bend right=20] (x3);
  \draw[->, very thick] (x4) to[bend right=20] (nx4);
  \draw[->, very thick] (nx4) to[bend right=20] (x4);
  \draw[->, very thick] (x5) to[bend right=20] (nx5);
  \draw[->, very thick] (nx5) to[bend right=20] (x5);
  
  \draw[->] (x1) to[bend right=10] (S);
  \draw[->] (nx1) to[bend right=10] (S);
  \draw[->] (x2) to[bend right=10] (S);
  \draw[->] (nx2) to[bend right=10] (S);
  \draw[->] (x3) to[bend right=10] (S);
  \draw[->] (nx3) to[bend right=10] (S);
  \draw[->] (x4) to[bend left=10] (S);
  \draw[->] (nx4) to[bend left=10] (S);
  \draw[->] (x5) to[bend left=10] (S);
  \draw[->] (nx5) to[bend left=10] (S);
  \label{fig-demand-clause}
 
  \end{scope}

  \begin{scope}[rotate=-90, xshift=5cm]
   \node[circle,draw,inner sep=3pt,label=$C_1$] at(0,0) (C) {};
\node[circle,draw,inner sep=3pt,label=$C_2$] at(0,2) (C2) {};

\node[circle,draw,inner sep=3pt,label={[yshift=4pt, xshift=3pt]left:${x_5}$}] at(2,3.7) (x5) {};
\node[circle,draw,inner sep=3pt,label={[yshift=4pt, xshift=3pt]left:${x_4}$}] at(2,1.5) (x4) {};
  \node[circle,draw,inner sep=3pt,label={[yshift=4pt, xshift=3pt]left:${x_3}$}] at(2,-0.4) (x3) {};
  \node[circle,draw,inner sep=3pt, label={[yshift=4pt, xshift=3pt]left:${x_2}$}] at (2,-2.5) (x2) {};
  \node[circle,draw,inner sep=3pt,label={[yshift=4pt, xshift=3pt]left:$x_1$}] at(2,-4.6) (x1) {};

  \node[circle,draw,inner sep=3pt,label={[xshift=-3pt, yshift=4pt]right:$\overline{x}_5$}] at(2,4.6) (nx5) {};
  \node[circle,draw,inner sep=3pt,label={[xshift=-3pt, yshift=4pt]right:$\overline{x}_4$}] at(2,2.5) (nx4) {};
  \node[circle,draw,inner sep=3pt,label={[xshift=-3pt, yshift=4pt]right:$\overline{x}_3$}] at(2,0.4) (nx3) {};
  \node[circle,draw,inner sep=3pt,label={[xshift=-3pt, yshift=4pt]right:$\overline{x}_2$}] at(2,-1.5) (nx2) {};
  \node[circle,draw,inner sep=3pt,label={[xshift=-3pt, yshift=4pt]right:$\overline{x}_1$}] at(2,-3.6) (nx1) {};

  \node[circle,inner sep=-1pt] at(4,1) (S) {\specialstar};

  \draw[->] (C)to[bend right=10] node[pos=0.5] {\includegraphics[width=4mm]{pigeon.png}} (x1);
  \draw[->] (C) to[bend right=10] node[pos=0.5] {\includegraphics[width=4mm]{pigeon.png}} (x2);
  \draw[->] (C) to[bend left=10] node[pos=0.5] {\includegraphics[width=4mm]{pigeon.png}} (nx3);
  \draw[->] (C2)to[bend right=10] node[pos=0.4] {\includegraphics[width=4mm]{pigeon.png}} (x3);
  \draw[->] (C2)to[bend right=10] node[pos=0.5] {\includegraphics[width=4mm]{pigeon.png}} (x4);
  \draw[->] (C2) to[bend left=10] node[pos=0.5] {\includegraphics[width=4mm]{pigeon.png}} (x5);

  \draw[->] (x1) to[bend right=20] node[pos=0.6] {\includegraphics[width=4mm]{pigeon.png}} (nx1);
  \draw[->] (nx1) to[bend right=20] node[pos=0.6] {\includegraphics[width=4mm]{pigeon.png}} (x1);
  \draw[->] (x2) to[bend right=20] node[pos=0.6] {\includegraphics[width=4mm]{pigeon.png}} (nx2);
  \draw[->] (nx2) to[bend right=20] node[pos=0.6] {\includegraphics[width=4mm]{pigeon.png}} (x2);
  \draw[->] (x3) to[bend right=20] node[pos=0.6] {\includegraphics[width=4mm]{pigeon.png}} (nx3);
  \draw[->] (nx3) to[bend right=20] node[pos=0.6] {\includegraphics[width=4mm]{pigeon.png}} (x3);
  \draw[->] (x4) to[bend right=20] node[pos=0.6] {\includegraphics[width=4mm]{pigeon.png}} (nx4);
  \draw[->] (nx4) to[bend right=20] node[pos=0.6] {\includegraphics[width=4mm]{pigeon.png}} (x4);
  \draw[->] (x5) to[bend right=20] node[pos=0.6] {\includegraphics[width=4mm]{pigeon.png}} (nx5);
  \draw[->] (nx5) to[bend right=20] node[pos=0.6] {\includegraphics[width=4mm]{pigeon.png}} (x5);
  
  \draw[->] (x1) to[bend right=10] node[pos=0.5] {\includegraphics[width=4mm]{pigeon.png}} (S);
  \draw[->] (nx2) to[bend right=10] node[pos=0.5] {\includegraphics[width=4mm]{pigeon.png}} (S);
  \draw[->] (x3) to[bend right=10] node[pos=0.5] {\includegraphics[width=4mm]{pigeon.png}} (S);
  \draw[->] (nx4) to[bend left=10] node[pos=0.5] {\includegraphics[width=4mm]{pigeon.png}} (S);
  \draw[->] (nx5) to[bend left=10] node[pos=0.5] {\includegraphics[width=4mm]{pigeon.png}} (S);
  \label{fig-infrastructure-clause}
  \end{scope}

\end{tikzpicture}
}
\caption{The constructed demand graph for $\varphi = (x_1 \vee \overline{x}_3 \vee x_2) \wedge (x_3\vee x_4\vee x_5)$ without forced edge gadgets is illustrated above. Demand edges are added from $C_1$ and $C_2$ to their three literal nodes and to the star, from $x_j$ to $\overline{x}_j$ and from~$\overline{x}_j$ to~$x_j$ for all $j\in [n]$, and from each literal ($x_j$ or~$\overline{x}_j$) to the star. Forced edges are depicted with thick lines. Below is the corresponding infrastructure graph for the assignment $(x_1, x_2, x_3, x_4, x_5)=(\textsc{true}, \textsc{false}, \textsc{true}, \textsc{false}, \textsc{false})$.}
\label{fig:clause-gadget}
\end{figure}

($\Rightarrow$) Assume that~$\varphi$ is a yes-instance, that is, there exists a truth assignment to the variables of $\varphi$ such that all clauses are satisfied. Now, we build a solution for the constructed instance $(G_D,k)$ of \thp. We place the pigeons in the following way: three pigeons in each clause $C_i$ with home node literal $x_j$ that appears in that clause (total $3m$ pigeons); one pigeon in literal $x_j$ with home in $\overline{x}_j$ and one pigeon in $\overline{x}_j$ with home in $x_j$ for all $j\in [n]$ (total $2n$ pigeons); one pigeon in $x_j$ or $\overline{x}_j$ with home node \specialstar based on the truth assignment in $\varphi$ for all $j\in [n]$ (total $n$ pigeons). 
Additionally, for each forced edge $e=(a,b)$ and each~$i \in [2n+4]$, we add a pigeon in $u_{i,1}^e$ with home node $u_{i,2}^e$; one pigeon in $u_{i,2}^e$ with home node $u_{i,3}^e$; one pigeon in $u_{i,3}^e$ with home node $a$ (total $6n+12$ pigeons per forced edge).
This construction is shown in Figure \ref{fig:forced-edge-gadget}. Note that the total number of pigeons placed is~$3m + 2n + n + (2n+3m)(6n+12) = 12n^2+18mn+27n+39m= k$. 

\begin{figure}[t]
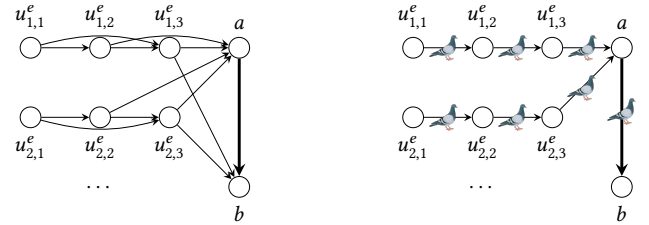

\centering
\resizebox{\columnwidth}{!}{%
\begin{tikzpicture}[> = stealth]
\begin{scope}
 \node[circle,draw,inner sep=3pt,label=above:$a$] at(0,-1) (a) {};
  \node[circle,draw,inner sep=3pt,label=below:$b$] at(0,-3) (b) {};

  \node[circle,draw,inner sep=3pt,label=above:$u_{1,1}^{e}$] at(-3,-1) (u1) {};
  \node[circle,draw,inner sep=3pt,label=above:$u_{1,2}^{e}$] at(-2,-1) (u2) {};
  \node[circle,draw,inner sep=3pt,label=above:$u_{1,3}^{e}$] at(-1,-1) (u3) {};
  \node[circle,draw,inner sep=3pt,label=below:$u_{2,1}^{e}$] at(-3,-2) (u4) {};
  \node[circle,draw,inner sep=3pt,label=below:$u_{2,2}^{e}$] at(-2,-2) (u5) {};
  \node[circle,draw,inner sep=3pt,label=below:$u_{2,3}^{e}$] at(-1,-2) (u6) {};
  \node at (-2,-3) {$\dots$}; 
 
  \draw[->] (u1) -- (u2);
  \draw[->] (u1) to[bend left=15] (u3);
  \draw[->] (u2) -- (u3);
  \draw[->] (u2) to[bend left=15] (a);
  \draw[->] (u3) -- (a);
  \draw[->] (u3) -- (b);

  \draw[->] (u4) -- (u5);
  \draw[->] (u4) to [bend right=15] (u6);
  \draw[->] (u5) -- (u6);
  \draw[->] (u5) -- (a);
  \draw[->] (u6) -- (a);
  \draw[->] (u6) -- (b);

  \draw[->, very thick] (a) -- (b);

  \end{scope}

  \begin{scope}[xshift=5.5cm] 
  \node[circle,draw,inner sep=3pt,label=above:$a$] at(0,-1) (a) {};
  \node[circle,draw,inner sep=3pt,label=below:$b$] at(0,-3) (b) {};
  
  \node[circle,draw,inner sep=3pt,label=above:$u_{1,1}^{e}$] at(-3,-1) (u1) {};
  \node[circle,draw,inner sep=3pt,label=above:$u_{1,2}^{e}$] at(-2,-1) (u2) {};
  \node[circle,draw,inner sep=3pt,label=above:$u_{1,3}^{e}$] at(-1,-1) (u3) {};
  \node[circle,draw,inner sep=3pt,label=below:$u_{2,1}^{e}$] at(-3,-2) (u4) {};
  \node[circle,draw,inner sep=3pt,label=below:$u_{2,2}^{e}$] at(-2,-2) (u5) {};
  \node[circle,draw,inner sep=3pt,label=below:$u_{2,3}^{e}$] at(-1,-2) (u6) {};
  \node at (-2,-3) {$\dots$}; 
 
  \draw[->] (u1) to node[pos=0.4] {\includegraphics[width=4mm]{pigeon.png}} (u2);
  \draw[->] (u2) to node[pos=0.4] {\includegraphics[width=4mm]{pigeon.png}} (u3);
  \draw[->] (u3) to node[pos=0.4] {\includegraphics[width=4mm]{pigeon.png}} (a);
  
  \draw[->] (u4) to node[pos=0.4] {\includegraphics[width=4mm]{pigeon.png}} (u5);
  \draw[->] (u5) to node[pos=0.4] {\includegraphics[width=4mm]{pigeon.png}} (u6);
  \draw[->] (u6) to node[pos=0.4] {\includegraphics[width=4mm]{pigeon.png}} (a);
  
  \draw[->, very thick] (a) to node[pos=0.45] {\includegraphics[width=4mm]{pigeon.png}} (b);

\end{scope}

\end{tikzpicture}
}
\caption{Two arms of a forced edge gadget in the demand graph are illustrated on the left. Edge~$e=(a,b)$ is forced and depicted with a thick line. On the right, a corresponding infrastructure graph that satisfies the forced edge gadget from the demand graph with the forced edge and 3 pigeons per arm is shown.}
\label{fig:forced-edge-gadget}
\end{figure}

We next discuss how to schedule the pigeons.
First, we let the pigeons fly from~$u_{i,1}^e$ to~$u_{i,2}^e$, then~$u_{i,2}^e$ to~$u_{i,3}^e$, and finally~$u_{i,3}^e$ to $a$ for each forced edge~$e=(a,b)$ and each~$i \in [2n+4]$. Next, pigeons fly from each clause $C_i$ to the literals appearing in that clause. Then, pigeons between the literals $x_j$ and $\overline{x}_j$ for all $j\in [n]$. Finally, for every variable $x_j$, a pigeon is released from $x_j$ to \specialstar if~$x_j$ is assigned true and a pigeon is released from~$\overline{x}_j$ to \specialstar if it is assigned false.

Now, we show that each demand is satisfied either directly or over two hops using an intermediate node in the constructed infrastructure graph. In particular, whenever a pigeon is placed on an edge $(u,v)$, the corresponding demand $(u,v)$ is satisfied by this direct flight.
The only demands that are not satisfied in this way are~$(u_{i,1}^e,u_{i,3}^e)$, $(u_{i,2}^e,a)$, $(u_{i,3}^e,b)$ for each forced edge~$e=(a,b)$ and each~$i \in [2n+4]$, $(C_i,\specialstar)$ for each clause~$C_i$, and for each variable~$x_j$ either~$(x_j,\specialstar)$ (if~$x_j$ is set to false) or~$(\overline{x}_j,\specialstar)$ (if~$x_j$ is set to true).
The demands within each forced edge gadget are satisfied over 2-hop paths with pigeons flying~$u_{i,1}^e \to u_{i,2}^e \to u_{i,3}^e$; $u_{i,2}^e \to u_{i,3}^e \to a$; and~$u_{i,3}^e \to a \to b$, respectively. 
The demand from each literal $x_j$ assigned false (or~$\overline{x}_j$ where~$x_j$ is set to true) to \specialstar is satisfied over a path $x_j \to \overline{x}_j \to \specialstar$ (or~$\overline{x}_j \to x_j \to \specialstar$). Finally, the demand $(C_i, \specialstar)$ is satisfied with a 2-hop flight $C_i \to x_j \to \specialstar$ where $x_j$ is a variable appearing in $C_i$ that is assigned true (or~$\overline{x}_j$ and~$x_j$ is set to false). Note that this variable must exist as otherwise $\varphi$ would not be satisfied, which contradicts our assumption.

($\Leftarrow$) Now, assume that~$(G_D,(11n+15m))$ is a yes-instance of \thp.
We will show that each forced edge gadget can be assigned a score of~$6n+13$ in any solution where each pigeon gives a total score of at most~$1$.
Hence, a total score of at most~$k - (2n+3m)(6n+13) = n$ can be unassigned and no forced edge gadget can be assigned more than~$7n+13$ as otherwise the solution contains more than~$k$ pigeons.
To define the score of a gadget, we distinguish between public and private nodes.
A node~$u_{i,j}^e$ for some forced edge~$e$ and any~$i,j$ is called a private node and all remaining nodes ($C_i,x_j,\overline{x}_j,\specialstar$) are called public.
First, for each private node~$v$, note that at least one pigeon has to leave from~$v$ since~$v$ is a source for some demand.
We arbitrarily assign one of the pigeons leaving~$v$ (a score of~$1$) to~$v$ and the score of any private node~$v = u_{i,j}^e$ is also assigned the forced edge gadget for~$e$.
Afterwards, all unassigned pigeons are assigned as follows: If the pigeon flies from~$u$ to~$v$ where both~$u$ and~$v$ are public nodes, then no node is assigned a score, but if~$(u,v)$ is a forced edge, then a score of one is added to the forced edge gadget for~$(u,v)$.
If exactly one of the nodes~$u$ and~$v$ is a public node and the other is a private node, then a score of one is assigned to the private node.
Finally, if both are public nodes, then a score of~$\frac{1}{2}$ is added to each of the two nodes.
Note that each pigeons gives a total score of at most one to all nodes and also a total score of at most one to each forced edge gadget.
Moreover, since each private node is assigned a score of at least one, each forced edge gadget is assigned at least a score of~$6n+12$.

Let~$e=(a,b)$ be a forced edge, we show that if no pigeon flies from~$a$ to~$b$ in the solution, then each arm of the forced edge gadget is assigned a total score of at least~$3.5$.
Since the number of arms is~$2n+4 \geq 2$, each forced edge gadget is thus assigned a score of at least~$6n+13$.
Moreover, if no pigeon flies from~$a$ to~$b$, then the forced edge gadget for~$e$ is assigned a score of at least~$3.5 (2n+4) = 7n+14$.
As argued above, this means that the number of pigeons in the solution is larger than~$k$, a contradiction.
So now assume towards a contradiction that any arm~$\{u_{i,1}^e,u_{i,2}^e,u_{i,3}^e\}$ is assigned a total score of less than~$3.5$ and no pigeon flies from~$a$ to~$b$.
Since each pigeon assigns a score of~$\frac{1}{2}$ or~$1$ to a node and each node is assigned a score of at least one, the three nodes in the arm are assigned a score of exactly three and each node in the arm is assigned a score of exactly one.
Note that this implies that each node has exactly one pigeon in the solution that leaves the respective node.
No pigeon flies from~$u_{i,3}^e$ to~$a$ as no pigeon flies from~$a$ to~$b$ in the solution and thus the demand~$(u_{i,3}^e,b)$ cannot be satisfied by a 2-hop path.
For the same reason, no pigeon flies from~$u_{i,2}^e$ to~$u_{i,3}^e$ as the demand~$(u_{i,2}^e,a)$ could not be satisfied and no pigeon flies from~$u_{i,1}^e$ to~$u_{i,2}^e$ as the demand~$(u_{i,1}^e,u_{i,3}^e)$ could not be satisfied.
Let~$c$ be the node such that a pigeon flies from~$u_{i,1}^e$ to~$c$ in the solution.
We make a case distinction whether~$c=u_{i,3}^e$ or not.
If~$c \neq u_{i,3}^e$, then note that one pigeon has to fly from~$c$ to~$u_{i,2}^e$ and one pigeon has to fly from~$c$ to~$u_{i,3}^e$. At most one of these pigeons is fully assigned to~$c$ and the other adds a score of at least~$\frac{1}{2}$ to the three nodes in the arm, raising the total score to at least~$3.5$.
So assume that~$c=u_{i,3}^e$.
To satisfy the demand~$(u_{i,1}^e,u_{i,2}^e)$, one pigeon has to fly from~$u_{i,3}^e$ to~$u_{i,2}^e$.
Then, a pigeon has to fly from~$u_{i,2}^e$ to both~$a$ and~$b$, contradicting that exactly one pigeon leaves each of the three nodes in the arm.
This shows that if no pigeon flies from~$a$ to~$b$, then the total score of each arm is at least~$3.5$ and the number of pigeons in the solution is larger than~$k$.
Finally, note that since we may assume that a pigeon flies from~$a$ to~$b$ for each forced edge~$e=(a,b)$, it holds that the score assigned to the forced edge gadget for~$e$ is at least~$6n+13$ and each pigeon that flies from a public node to a private node in the gadget increases the score by one.

To conclude the proof, consider the set of all pigeons in a solution that do not fly from~$a$ to~$b$ for some forced edge~$(a,b)$ and that do not start in a private node.
These are all but at least~$(2n+3m)(6n+13) = k-n$.
Hence, these are at most~$n$ pigeons.
Assume towards a contradiction that for some variable~$x_j$, none of these pigeons starts at~$x_j$ or~$\overline{x}_j$.
Then, the only pigeons leaving~$x_j$ and~$\overline{x}_j$ are the pigeons flying the forced edges~$(x_j,\overline{x}_j)$ and~$(\overline{x}_j,x_j)$.
Thus, the demands~$(x_j,\specialstar)$ and~$(\overline{x_j},\specialstar)$ are not satisfied, a contradiction.
Thus for each variable~$x_j$, exactly one of the~$n$ pigeons flies from~$x_j$ or from~$\overline{x}_j$.
We define a truth assignment by setting $x_j=\textsc{true}$ if $x_j$ is the node that the additional pigeon flies from and to $\textsc{false}$ otherwise.

It remains to show that the constructed assignment satisfies~$\varphi$.
To this end, consider any clause~$C_i$ and the demand~$(C_i,\specialstar)$.
As no pigeons are left to fly from~$C_i$ to~$\specialstar$ directly and the only pigeons flying from~$C_i$ are to the three nodes representing literals that appear in~$C_i$ (the three forced edges incident to~$C_i$), it must hold that a pigeon flies from one of these three nodes to~$\specialstar$.
These has to be one of the~$n$ pigeons that do not fly forced edges and do not start in a private node.
Hence, at least one of the three nodes is assigned an additional pigeon and by construction, the assignment of that variable satisfies~$C_i$.
Since all variables are satisfied in this way be the assignment, $\varphi$ is satisfied and the original instance of \textsc{3SAT} is a yes-instance.
This concludes the proof.

\end{proof}

In the following, we will present an ILP formulation of the \thp{} problem. To simplify the formulation, we make use of the following lemma:
\begin{lemma}
    \label{lem:numberhops}
    Any
    optimal solution for both \mhp{} and \thp{} requires at most $2(n-1)$ pigeons. 
\end{lemma}
\begin{proof}
    We show this statement by presenting a solution that uses $2(n-1)$ pigeons for any demand graph. We build an infrastructure graph as a directed cycle, where the nodes along the cycle are following some arbitrary order. Assume WLOG that the nodes are ordered as $v_1,v_2,\ldots, v_n$ along the cycle. We place pigeons at the nodes as follows: two pigeons at each node $v_i, i\in[n-2]$ with a home at $v_{i+1}$, one pigeon at $v_{n-1}$ with a home at $v_n$, and one pigeon at $v_n$ with a home at $v_1$. We release the pigeons sequentially, one at a time, starting with node $v_1$. This pigeon carries the demand from $v_1$ to all other nodes in the graph. Once the pigeon arrives at $v_2$, the remaining demand from $v_1$ is forwarded with the next pigeon, together with the demand from $v_2$ to all other nodes. After making one cycle through all nodes, node $v_1$ will receive a pigeon with demands from nodes ${v_2,\ldots,v_n}$ with destinations in ${v_1,\ldots,v_{n-1}}$. Observe that at this point, the demand of node $v_1$ has been delivered to all nodes. However, the demand of node $v_n$ has only been delivered to node $v_1$ so far. Therefore, we need to continue forwarding the remaining demand through the cycle, up to node $v_{n-1}$. This construction always forwards all the demand with $2n-2$ pigeons.
\end{proof}
The following theorem 
establishes the properties of
an ILP we present for the \thp{} problem.
\begin{theorem}
    \label{thm:ILPtp}
    There is an ILP with~$O(n^4)$ variables and~$O(n^3)$ constraints for \thp, where all variables are binary.
\end{theorem}

\begin{proof}
    Let~$(G=(V,E),k)$ be an instance of \thp.
    We construct an ILP with a binary variable~$x_{u,v}^i$ for each pair~$u,v \in V$ and each~$i \in [2n-2]$ and a binary variable~$y_{u,w,v}^{i}$ for each edge~$(u,v) \in E$, each node~$w \in V$, and each~$i \in [2n-2]$.
    The variable~$x_{u,v}^i$ is set to true if the~$i$\textsuperscript{th} pigeon flies from~$u$ to~$v$.
    Note that by Lemma~\ref{lem:numberhops}, we may assume that at most~$2n-2$ pigeons are required.
    The goal is to minimize the number of pigeons, that is,~$\sum_{u,v \in V} \sum_{i \in [2n-2]}x_{u,v}^i$, and a solution with~$k$ pigeons will exist if and only if the goal value is at most~$k$.
    The constraints are listed in Figure~\ref{fig:ILPtwo}.
\begin{figure}[t]
\centering
\begin{align}
\sum_{u,v \in V} x_{u,v}^i &\le 1 
&\hfill \forall i\in [2n-2] \label{req:4}\\
\sum_{i \in [2n-2]} (x_{u,v}^i + \sum_{w \in V} y_{u,w,v}^i)
&\ge 1 
&\hfill \forall (u,v)\in E \label{req:5}\\
y_{u,w,v}^{i} &\le x_{u,w}^i 
&\hfill \forall (u,v)\in E,\ i\in [2n-2] \label{req:6}\\
y_{u,w,v}^i 
&\le \sum_{\substack{j\in [2n-2]\\ j>i}} x_{w,v}^j
&\hfill \forall (u,v)\in E,\ i\in [2n-2] \label{req:7}
\end{align}
\caption{The constraints of the ILP for \thp.}
\label{fig:ILPtwo}
\end{figure}

    Since the number of variables is clearly in~$O(n^4)$ and the numbers of constraints is in~$O(n^3)$, it remains to show that the ILP is correct.
    The variable~$y_{u,w,v}^i$ will be set to true if and only if the demand~$(u,v)$ is routed via node~$w$ and the~$i$\textsuperscript{th} pigeon transports the message from~$u$ to~$w$.

    To show correctness, first assume that there is a solution with~$k$ pigeons.
    We may assume without loss of generality that no two pigeons fly at the same time, and hence we can order all pigeons in the order they fly.
    Then, we set~$x_{u,v}^{i} = 1$ if and only if the~$i$\textsuperscript{th} pigeon flies from~$u$ to~$v$.
    For each demand~$(u,v)$, if any pigeon flies from~$u$ to~$v$ directly, then we set all variables~$y_{u,w,v}^i=0$.
    Otherwise, there is at least one node~$w$ such that a pigeon~$i$ flies from~$u$ to~$w$ and a later pigeon~$j$ that flies from~$w$ to~$u$.
    We arbitrarily chose one such pair~$w,i$ and set~$y_{u,w,v}^i$ to true.
    Note that requirements~\ref{req:4}, \ref{req:5}, and~\ref{req:6} are now all satisfied by construction.
    Moreover, since we assume that for each chosen pair~$(w,i)$ a later pigeon~$j$ flies from~$w$ to~$v$, also constraint~\ref{req:7} is satisfied.

    In the other direction, assume that there is a solution to the ILP where the goal value is at most~$k$.
    We will let a pigeon fly from a node~$u$ to node~$v$ in time step~$i$ if and only if~$x_{u,v}^i=1$.
    Note that by construction at most~$k$ pigeons fly (each at some time step between~$1$ and~$2n-2$).
    It now only remains to show that each demand is satisfied by the constructed solution.
    So consider any demand~$(u,v)$.
    If some pigeon flies from~$u$ to~$v$ directly, then the demand is satisfied.
    Otherwise, $x_{u,v}^i = 0$ for all~$i \in [2n-2]$.
    By constraint~\ref{req:5}, at least one variable~$y_{u,w,v}^i$ is set to true.
    By constraints~\ref{req:6} and~\ref{req:7}, $x_{u,w}^i$ and~$x_{w,v}^j$ are set to~$1$ for at least one node~$w$ and one value~$j \in [2n-2]$ with~$j > i$.
    Thus, the demand~$(u,v)$ is satisfied by the pigeon flying from~$u$ to~$w$ at time~$i$ and the pigeon flying from~$w$ to~$v$ at time~$j$.
    Since the demand was chosen arbitrarily, all demands are satisfied in this way, and the ILP is therefore correct.
\end{proof}

\subsection{\mhp{} Solution}

Observe that Algorithm~\ref{alg:twohop-coordinator} is also a ($2-\sum_{i\in[c]}\Delta(C_i)$)-approximation for the \mhp{} problem. 
We next show that also the \mhp{} problem is NP-hard, and provide an ILP for this problem. 

The proof of NP-hardness is based on the following three lemmas: 

\begin{lemma}
    \label{lem:disconnect}
    If the demand graph~$G=(V,E)$ is weakly disconnected, then 
    each weakly connected component of~$G$ can be solved independently for \mhp{} and the minimum number of pigeons required for~$G$ is the sum of the minimum numbers of pigeons required for each connected component of~$G$.
\end{lemma}
\begin{proof}
  Consider a graph~$G=(V,E)$ with two weakly connected components, denoted $C_1$ and $C_2$. Assume by means of contradiction that there is an optimal solution where a pigeon $p$ flies from $C_1$ to $C_2$. Since $C_1$ and $C_2$ are weakly disconnected, the destination nodes in $C_1$ and $C_2$ are disjoint. Consider first the case where $p$ has not carried any demand, then the solution was not optimal because it wasted a pigeon. Thus, pigeon $p$ must have carried some demand. WLOG, assume that this demand was from sources in $C_1$, and thus has destinations in $C_1$. But then, there must exist at least one other pigeon that carries the demand back from $C_2$ to $C_1$. In this case, however, we could have saved at least one pigeon by avoiding a detour over the nodes in $C_2$. This is a contradiction to the fact that we had an optimal solution.
\end{proof}

\begin{lemma}
    \label{lem:walk}
    If the demand graph~$G=(V,E)$ is weakly connected, then there is always an optimal solution for \mhp{} in which only one pigeon flies at a time and the pigeon starting at time~$i$ starts at the same node that the pigeon flying at time~$i-1$ ended their flight for all~$i > 1$.
\end{lemma}

\begin{proof}
    Let~$OPT$ be an optimal solution.
    We assume without loss of generality that the pigeon flights in~$OPT$ are strictly ordered. This is possible, as any two flights that happen at the same time cannot conflict each other and thus can be ordered arbitrarily.
    For simplicity, let $e_i=(a_i,b_i)$ denote the edge in the infrastructure graph that corresponds to the $i$-th pigeon flight in~$OPT$ and let~$OPT_i$ denote the partial solution only containing the first~$i$ flights.

    For each~$0 \leq i \leq |OPT|$, we build a sequence~$\sigma_i$ of pigeon flights satisfying the following two conditions.
    First, all pigeon flights in a weakly connected component of the infrastructure graph corresponding to~$\sigma_i$ occur consecutive and the starting node of each pigeon is the destination of the previous pigeon unless it is the first flight in a weakly connected component.
    Second, for each pair~$(u,v)$ of nodes (not necessarily terminals), if a pigeon route from~$u$ to~$v$ exists in~$OPT_i$, then it also exists in~$\sigma_i$, that is, each node has at least as much information after pigeons flew according to~$\sigma_i$ as for~$OPT_i$.
    Note that if we succeed with this construction for all~$i$, then the final sequence~$\sigma_{|OPT|}$ is an optimal solution satisfying the requirements of the lemma.

    Let~$\sigma_0$ be the empty sequence.
    Note that $\sigma_0$ fulfills our requirements.
    Now assume that~$\sigma_{i-1}$ fulfills our two requirements.
    We will show how to construct~$\sigma_{i}$.
    Consider the $i$-th pigeon flight~$e_{i}=(a_i,b_i)$.
    We consider two cases: either both~$a_i$ and~$b_i$ belong to the same weakly connected component of the infrastructure graph corresponding to~$\sigma_{i-1}$ or not.
    If both belong to the same weakly connected component, then consider the subsequence of~$\sigma_{i-1}$ of all flights in this component and let~$c$ be the destination of the last edge in this sequence.
    Now insert the flight~$(c,b_i)$ to the end of the subsequence (and shift all flights that appear later in~$\sigma_{i-1}$ one time slot back).
    Note that the length of the sequence increased by exactly one.
    Moreover, the first requirement is met since~$e_i$ only contains vertices from one weakly connected component and the new flight starts at the last node of the previous subsequence.
    For the second requirement, note that all nodes who have a pigeon tour to~$a_i$ in~$\sigma_{i-1}$ also have a pigeon route to~$c$.
    Thus, the pigeon flight~$(c,b_i)$ ensures that node~$b_i$ learns at least as much information in~$\sigma_i$ as it does in~$OPT_i$.

    Now consider the case where~$e_i = (a_i,b_i)$ connects two weakly connected components of the infrastructure graph corresponding to~$\sigma_{i-1}$.
    We build~$\sigma_i$ as follows.
    First, let~$\pi_a$ and~$\pi_b$ be the subsequences of~$\sigma_{i-1}$ of all edges in the weakly connected components containing~$a_i$ and~$b_i$, respectively.
    Let~$c$ be the destination of the last edge in~$\pi_a$ and let~$d$ be the first starting point of an edge in~$\pi_b$.
    Now, we remove both~$\pi_a$ and~$\pi_b$ from~$\sigma_{i-1}$ and instead add the sequence~$\pi_a \circ ((c,d)) \circ \pi_b$ at the end.
    Note that~$\sigma_i$ is longer than~$\sigma_{i-1}$ by exactly one and thus has size~$i=|OPT_i|$.
    The first requirement is met by construction and it remains to show that the second requirement is also met.
    To this end, note that~$\pi_a$ and~$\pi_b$ are disjoint and hence no node in either weakly connected component has received any message from a node in the other component.
    Adding~$e_i$ in~$OPT$ hence only gives~$b_i$ information about all nodes that~$a_i$ knows about.
    Note that~$c$ has knowledge about all nodes appearing in~$\pi_a$ in~$\sigma_{i-1}$ and therefore also in~$\tau_a$.
    Moreover, there is a pigeon flight from~$d$ to~$b_i$ in~$\tau_b$ and hence~$b_i$ has full knowledge about all nodes in the component of~$a_i$ in~$\sigma_i$.
    Thus, the second requirement is met and this concludes the proof.
\end{proof}

Using these lemmas, we can now show the hardness of the \mhp{} problem. Note that this result is incomparable to Theorem~\ref{thm:red-2spanner-to-thp} and neither result immediately implies the other.

\begin{theorem}
    \label{thm:multihard}
    The decision variant of \mhp{} is NP-hard.
\end{theorem}

\begin{proof}
    We reduce from \textsc{Vertex Cover}.
    To this end, let~$(G=(V,E),k)$ be an instance of \textsc{Vertex Cover} and assume without loss of generality, that~$G$ is connected and contains at least one edge.
    Note that this implies that each node in~$V$ is incident to at least one edge in~$E$.
    We will construct an equivalent instance~$(G'=(V,E'),k')$ of \mhp{} on the same node set~$V$ as follows.
    For each edge~$\{u,v\} \in E$, we add  the two edges~$(u,v)$ and~$(v,u)$ to~$E'$.
    To conclude the construction, we set~$k'=n+k-1$, where~$n=|V|$.
    See Figure~\ref{fig:exampleVC} for an example of the above construction.

    \begin{figure}[t]
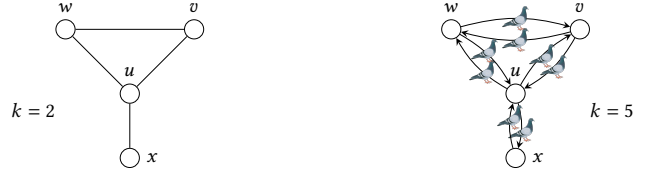

        \centering
        \resizebox{\columnwidth}{!}{%
        \begin{tikzpicture}[> = stealth]
            \node[circle,draw,inner sep=3pt,label=$u$] at(0,0) (u) {};
            \node[circle,draw,inner sep=3pt,label=$v$] at(1,1) (v) {};
            \node[circle,draw,inner sep=3pt,label=$w$] at(-1,1) (w) {};
            \node[circle,draw,inner sep=3pt,label=right:$x$] at(0,-1) (x) {};
            
            \node[circle,draw,inner sep=3pt,label=$u$] at(6,0) (u2) {};
            \node[circle,draw,inner sep=3pt,label=$v$] at(7,1) (v2) {};
            \node[circle,draw,inner sep=3pt,label=$w$] at(5,1) (w2) {};
            \node[circle,draw,inner sep=3pt,label=right:$x$] at(6,-1) (x2) {};

            \node at(-1.5,-.25) {$k=2$};
            \node at(7.5,-.25) {$k=5$};

            \draw (u) -- (v);
        \draw (v) -- (w);
        \draw (u) -- (w);
        \draw (u) -- (x);

        \draw[->,bend left=15] (u2) to  node[pos=0.5] {\includegraphics[width=4mm]{pigeon.png}} (v2);
        \draw[->,bend left=15] (v2) to  node[pos=0.5] {\includegraphics[width=4mm]{pigeon.png}} (u2);

        \draw[->,bend left=15] (v2) to  node[pos=0.5] {\includegraphics[width=4mm]{pigeon.png}} (w2);
        \draw[->,bend left=15] (w2) to  node[pos=0.5] {\includegraphics[width=4mm]{pigeon.png}} (v2);

        \draw[->,bend left=15] (u2) to  node[pos=0.4] {\includegraphics[width=4mm]{pigeon.png}} (w2);
        \draw[->,bend left=15] (w2) to  node[pos=0.4] {\includegraphics[width=4mm]{pigeon.png}} (u2);

        \draw[->,bend left=15] (u2) to  node[pos=0.6] {\includegraphics[width=4mm]{pigeon.png}} (x2);
        \draw[->,bend left=15] (x2) to  node[pos=0.7] {\includegraphics[width=4mm]{pigeon.png}} (u2);
        \end{tikzpicture}
        }
        \caption{An example instance of \textsc{Vertex Cover} on the left and the corresponding equivalent instance of \mhp{} on the right.}
        \label{fig:exampleVC}
    \end{figure}

    Since the construction can clearly be computed in polynomial time, it only remains to prove that the constructed instance of \mhp{} is a yes-instance if and only if the original instance of \textsc{Vertex Cover} is a yes-instance.
    For this, first assume that the original instance of \textsc{Vertex Cover} is a yes-instance and let~$K$ be a vertex cover of size at most~$k$ in~$G$.
    Let~$s \in K$ be an arbitrary node 
    and let~$\pi = (u_1,u_2,\ldots,u_{|K|})$ 
 be an arbitrary ordering of the nodes in~$K$ where~$s$ is the last node and let~$\tau = (v_1,v_2,\ldots,v_{n-|K|})$ be an arbitrary ordering of the nodes in~$V \setminus K$.

    We now construct a solution for the constructed instance of \mhp.
    First, we place one pigeon in each node~$v_i$ whose home is node~$v_{i+1}$ for all~$i < n-|K|$.
    We also place one pigeon in~$s = u_{|K|}$ whose home is~$v_1$ and one pigeon in~$v_{n-|K|}$ whose home is~$u_1$.
    Finally, we place two pigeons in each node~$u_i$ with~$1 \leq i < |K|$ whose home is~$u_{i+1}$.
    Note that we placed exactly~${n+|K|-1 \leq n+k-1 = k'}$ pigeons.
    Next, we iteratively let one pigeon fly at a time as follows.
    We start with a pigeon flying from~$u_1$ to~$u_2$.
    Then a pigeon flying from~$u_2$ to~$u_3$ and so on.
    Once a pigeon arrives at~$u_{|K|}=s$, the next pigeon flies from~$s$ to~$v_1$.
    Afterwards, a pigeon flies from~$v_1$ to~$v_2$, then from~$v_2$ to~$v_3$ and so on until a pigeon arrives at~$v_{n-|K|}$.
    Now, the pigeon starting in~$v_{n-|K|}$ flies to~$u_1$.
    From now on, the second pigeon in each node~$u_i$ iteratively flies to~$u_{i+1}$ for each~$i < |K|$.

    We will next show that each demand is satisfied by the above construction.
    To this end, we show that for each demand~$(u,v)$, there is a subsequence in~$\pi \circ \tau \circ \pi$ which starts in~$u$ and ends in~$v$, where~$\circ$ is the concatenation operation.
    By construction, any subsequence corresponds to a sequence of pigeon flights that start and end at the respective nodes and are performed in the order given by the sequence.
    Thus, the information starting in node~$u$ is moved via pigeons iteratively and finally reaches~$v$.
    
    Note that each demand~$(u,v)$ implies that there is an edge~$\{u,v\}$ in~$E$ and hence at least one of the two nodes is contained in~$K$.
    If~$u$ is contained in~$K$, then the sought-after subsequence starts in the first occurrence of~$\pi$ and ends in either~$\tau$ or the second occurrence of~$\pi$ depending on whether~$v$ is contained in~$V \setminus K$ or in~$K$.
    If only~$v$ is contained in~$K$, then the subsequence starts in~$\tau$ and ends in the second occurrence of~$\pi$.
    Since each node in~$K$ is contained in~$\pi$ and each node in~$V \setminus K$ is contained in~$\tau$ by definition, the subsequences exist and we have successfully constructed a solution.

    In the other direction, assume that the constructed instance of \mhp{} is a yes-instance.
    By Lemma~\ref{lem:walk}, we may assume that a solution describes a walk through~$G'$ as~$G'$ is connected since~$G$ is connected.
    Let~$K$ be the set that includes all nodes in which at least two pigeons are initially placed in a solution as well as the the destination of the last pigeon flight.  
    We will show that~$K$ is a vertex cover of size at most~$k$. 
    First, note that since each node in~$V$ is incident to at least one edge in~$E$, it also holds for each node~$u$ that there is a demand of the form~$(u,v) \in E'$.
    Hence, at least one pigeon needs to start in each node as otherwise this demand cannot be satisfied.
    Since the number of pigeons is at most~$k' = n+k-1$, it holds that~$|K| \leq k-1+1$ 
    (the last $+1$ comes from the last destination added to~$K$).
    Assume towards a contradiction that~$K$ is not a vertex cover.
    Then, there is an edge~$(u,v)$ such that neither~$ u$ nor~$ v$ is contained in~$K$.
    Then, both of~$u$ and~$v$ appear once in the walk described by the assumed solution, as each node appearing twice either ends the walk or has two outgoing edges, which corresponds to two pigeons starting there initially.

    Since both~$ u$ and~$ v$ appear once, one of the two occurs earlier than the other.
    Let us assume without loss of generality that~$u$ appears before~$v$.
    Then, the constructed demand~$(v,u)$ (recall that~$\{u,v\} \in E$ and hence~$(u,v) \in E'$ and~$(v,u) \in E'$) cannot be satisfied as whenever a pigeon leaves from~$v$, no pigeon arrives at~$u$ to deliver the message, a contradiction to the assumption that we started with a solution.
    Thus, $K$ is a vertex cover of size at most~$k$ and the original instance of \textsc{Vertex Cover} is a yes-instance.
    This concludes the proof.
\end{proof}

In the following theorem, we present an ILP formulation of the \mhp{} problem.

\begin{theorem}
    \label{thm:ILPmp}
    There is an ILP with~$O(n^4)$ variables and constraints for \mhp, where all variables are binary.
\end{theorem}


\begin{proof}
    Let~$(G=(V,E),k)$ be an instance of \mhp.
    By Lemma~\ref{lem:disconnect}, we may assume that~$G$ is connected as otherwise, we can solve each connected component independently.
    We construct an ILP with a binary variable~$x_v^i$ for each~$v \in V$ and each~$i \in [2n]$ and a binary variable~$y_{u,v}^{i,j}$ for each edge~$(u,v) \in A$ and each pair~$i,j \in [2n]$ with~$i < j$.
    The goal is to minimize~$\sum_{v \in V} \sum_{i \in [2n]}x_v^i$ and a solution with~$k$ pigeons will exist if and only if there is a solution to the ILP where the goal value is at most~$k+1$.
    The constraints are listed in Figure~\ref{fig:ILPmulti}.
    \begin{figure}[t]
    \begin{align}
         \sum_{v \in V} x_v^i &\leq 1 & \forall i\in [2n] \label{req:1}\\
         \sum_{i < j \in [2n]} y_{u,v}^{i,j} &= 1 & \forall (u,v)\in E \label{req:2}\\
         2y_{u,v}^{i,j} &\leq x_u^i + x_v^j & \forall (u,v) \in E \text{ and } i < j \in [2n]\label{req:3}
    \end{align}
    \caption{The constraints of the ILP for \mhp.}
    \label{fig:ILPmulti}
    \end{figure}

    Since the number of variables and constraints is clearly both in~$O(n^4)$, it remains to show that the ILP is correct.
    To this end, we use Lemma~\ref{lem:walk} and Lemma~\ref{lem:numberhops}.
    We may assume that a solution is described by a walk of length at most~$2n-1$.
    We simply represent this walk by its at most~$2n$ nodes.
    We set~$x_v^i = 1$ if and only if~$v$ appears in position~$i$ in this sequence.
    Note that Constraint~\ref{req:1} ensures that at most one node appears in each position, and the goal function describes the total number of nodes appearing in the sequence.
    This number minus one is then the required number of pigeons as we claimed.

    So it only remains to show that each demand is satisfied by a solution to the ILP.
    To this end, we use variable~$y_{u,v}^{i,j}$ to denote that demand~$(u,v)$ is picked up at time step~$i$ and delivered in time step~$j$.
    Constraint~\ref{req:2} ensures that each demand is satisfied in this way, and Constraint~\ref{req:3} ensures that the solution walk is at node~$u$ at time~$i$ and at node~$v$ at time~$j$.
    Note that whenever~$y_{u,v}^{i,j}$ is set to~$1$, then in order to satisfy Constraint~\ref{req:3}, both~$x_u^i$ and~$x_v^j$ have to be set to~$1$ as well.
    This concludes the proof.
\end{proof}

\section{Future Work}

We view this paper as an ``Introduction to Pigeon Post Theory'' because it opens the door to many research avenues. First, throughout the paper, we focused on unbounded capacity pigeons, and the algorithmic and hardness results were derived under this assumption. An immediate question is how these results change when pigeons have bounded capacities, that is, when each pigeon can carry only a limited amount of information. For the special case of unit capacities, our hardness results for the 2-hop and multihop models seem to remain similar. Whether tight approximation results persist for bounded capacities, or whether new phenomena arise, remains an interesting open problem.
Second, while the present work focuses on minimizing the number of pigeons, other cost measures are equally natural. For instance, we could consider every pigeon hop as a time step and try to reduce the time needed to deliver all messages. 
Additionally, one may consider the cost of transporting pigeons to their remote nodes, possibly allowing batching of deliveries to multiple nodes.
Studying such cost-aware variants would further connect the model to classical network design and facility-location problems.
Another promising direction concerns the dynamic demand. In this work, demands are assumed to be known in advance. Instead it would be interesting to study settings in which demands arrive over time, either according to a known process, stochastically, or adversarially. This naturally leads to online and competitive variants of the problem.
Finally, it would be interesting to investigate distributed algorithms in the pigeon model. Since sending a message consumes a pigeon and effectively removes a directed edge from the infrastructure, one may ask how to perform distributed computation while minimizing pigeon usage, or how to compute without disconnecting the network.

\begin{acks}
    This work was supported by the German Research Foundation (DFG), SPP 2378 (ReNO-2), 2025-2029.
\end{acks}

\balance
\bibliographystyle{ACM-Reference-Format}
\bibliography{bib}
\end{document}